\documentclass[11pt]{article}
\usepackage{amssymb}
\usepackage{amsthm}
\usepackage{amsmath}
\usepackage[dvips]{epsfig}
\usepackage{graphicx}
\usepackage{color}
\usepackage{url}

\usepackage[utf8]{inputenc}
\usepackage[left=1cm, right=2cm, top=2cm, bottom=2cm]{geometry}
\usepackage[english]{babel}
\usepackage{amsfonts}
\usepackage[numbers]{natbib}
\newtheorem{theorem}{\bf Theorem}

\newtheorem{remark}{\bf Remark}
\newtheorem{proposition}{\bf Proposition}

\newcommand{\PT}{{\cal PT}}

\newcommand{\ltz}{{l^2}}
\newcommand{\ltzc}{{l^2_c(\mathbb{Z})}}
\newcommand{\ltzz}{{l^2(\mathbb{Z})}}

\title{\bf Long-time stability of breathers \\ in Hamiltonian $\PT$-symmetric lattices}

\author{Alexander Chernyavsky$^{1}$ and Dmitry E. Pelinovsky$^{1,2}$ \\
{\small $^{1}$ Department of Mathematics, McMaster University, Hamilton, Ontario, L8S 4K1, Canada} \\
{\small $^{2}$ Department of Applied Mathematics, Nizhny Novgorod State Technical University, Nizhny Novgorod, Russia }}

\begin{document}

\maketitle

\begin{abstract}
We consider the Hamiltonian version of a $\PT$-symmetric lattice that describes dynamics of coupled pendula
under a resonant periodic force. Using the asymptotic limit of a weak coupling between the pendula,
we prove the nonlinear long-time stability of breathers (time-periodic solutions localized in the lattice)
by using the Lyapunov method. Breathers are saddle points of the extended energy function,
which are located between the continuous bands of positive and negative energy. Nevertheless,
we construct an approximate Lyapunov function and estimate its evolution on a long but finite time interval.
The nonlinear stability analysis becomes possible for the $\PT$-symmetric lattice only because of
the existence of a Hamiltonian structure.
\end{abstract}

\section{Introduction}

We consider the following system of amplitude equations
\begin{eqnarray}
\left\{ \begin{array}{l} i \frac{du_n}{dt} = \epsilon \left( v_{n+1} - 2 v_n + v_{n-1} \right) + i \gamma u_n + \Omega v_n +
2\left[ \left( 2|u_n|^2 + |v_n|^2 \right) v_n + u_n^2 \bar{v}_n \right], \\
i \frac{dv_n}{dt} = \epsilon \left( u_{n+1} - 2 u_n + u_{n-1} \right) - i \gamma v_n + \Omega u_n +
2\left[ \left( |u_n|^2 + 2 |v_n|^2 \right) u_n + \bar{u}_n v_n^2 \right],
\end{array} \right.
\label{PT-dNLS}
\end{eqnarray}
where $\{ u_n, v_n \}_{n \in \mathbb{Z}}$ are complex-valued amplitudes that depend on time $t \in \mathbb{R}$,
whereas $(\Omega,\gamma,\epsilon)$ are real-valued parameters arising in a physical context described below.
We assume $\Omega \neq 0$, $\gamma > 0$, and $\epsilon > 0$ throughout our work.

The system (\ref{PT-dNLS}) describes a one-dimensional chain of coupled pendula, which are connected by
torsional springs with the tension coefficient $\epsilon$
in the longitudinal direction $n \in \mathbb{Z}$. Each pair of coupled pendula
$(u_n,v_n)$ are hung on a common string with a periodically varying tension coefficient
propositional to $\gamma$. When the frequency of the periodic force is in $1:2$ resonance with the frequency of pendula
detuned by $\Omega$, the system of Newton's equations of motion has been shown in \cite{ChernPel} to reduce asymptotically
to the system of amplitude equations (\ref{PT-dNLS}). Similar systems of amplitude equations
were derived previously in a number of physically relevant applications \cite{barashenkov2014,BBPS,BSRW}.
Figure \ref{netpic} depicts schematically the chain of coupled pendula.

\begin{figure}[h!]
\centering
\includegraphics[width=0.4\textwidth]{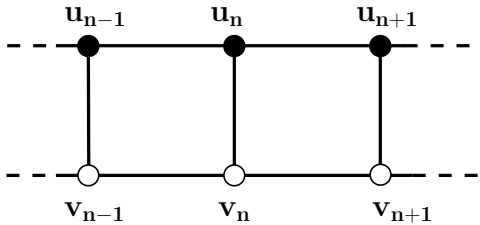}
\caption{The chain of $\PT$-symmetric dimers representing coupled pendula.
Filled (empty) circles correspond to sites with gain (loss).}
\label{netpic}
\end{figure}

The system of amplitude equations (\ref{PT-dNLS}) is usually referred to as the $\PT$-symmetric
discrete nonlinear Schr\"{o}dinger (dNLS) equation
because solutions to system (\ref{PT-dNLS}) remain invariant with respect to the
simultaneous action of the parity ($\mathcal{P}$) and time-reversal ($\mathcal{T}$) operators
given by
\begin{equation}
\label{PT-symmetry}
{\cal P} \left[ \begin{array}{c} u \\ v \end{array} \right] =
\left[ \begin{array}{c} v \\ u \end{array} \right], \qquad
{\cal T} \left[ \begin{array}{c} u(t) \\ v(t) \end{array} \right] =
\left[ \begin{array}{c} \bar{u}(-t) \\ \bar{v}(-t) \end{array} \right].
\end{equation}
Although the system of amplitude equations (\ref{PT-dNLS}) is autonomous, it is not norm-preserving and
the parameter $\gamma$ of the periodic tension coefficient represents the gain--loss coefficient. Indeed, if $\gamma > 0$,
the $\gamma$-term of the first equation induces the exponential growth of the amplitude $u_n$,
whereas the $\gamma$-term of the second equation induces the exponential decay of the amplitude $v_n$.
In spite of the gain-loss terms, many $\PT$-symmetric systems share properties of the Hamiltonian systems
and admit linearly stable zero equilibrium at least for sufficiently small values of $\gamma$ \cite{Bender}.

The remarkable property of the $\PT$-symmetric dNLS equation (\ref{PT-dNLS}) is the existence of the
cross--gradient symplectic structure \cite{barashenkov,barashenkov2014} with two conserved quantities
bearing the meaning of the energy and charge functions. Indeed, the $\PT$-symmetric dNLS
equation (\ref{PT-dNLS}) can be cast in the Hamiltonian form via the cross-gradient symplectic structure
\begin{equation}
i \frac{d u_n}{dt} = \frac{\partial H}{\partial \bar{v}_n}, \quad
i \frac{d v_n}{dt} = \frac{\partial H}{\partial \bar{u}_n}, \quad n \in \mathbb{Z},
\label{Hamiltonian-PT-dNLS}
\end{equation}
where the energy function is
\begin{eqnarray}
\nonumber
H & = & \sum_{n \in \mathbb{Z}} (|u_n|^2 + |v_n|^2)^2 + (u_n \bar{v}_n + \bar{u}_n v_n)^2 + \Omega (|u_n|^2 + |v_n|^2) \\
& \phantom{t} & - \epsilon |u_{n+1}-u_n|^2 - \epsilon |v_{n+1}-v_n|^2 + i \gamma ( u_n \bar{v}_n - \bar{u}_n v_n).
\label{energy-PT-dNLS}
\end{eqnarray}
The Hamiltonian system (\ref{Hamiltonian-PT-dNLS}) has an additional gauge symmetry,
with respect to the transformation $\{ u_n, v_n \}_{n \in \mathbb{Z}} \to \{ e^{i \alpha} u_n, e^{i \alpha} v_n \}_{n \in \mathbb{Z}}$,
where $\alpha \in \mathbb{R}$. The charge function related to the gauge symmetry is written in the form
\begin{equation}
\label{charge-PT-dNLS}
Q = \sum_{n \in \mathbb{Z}} (u_n \bar{v}_n + \bar{u}_n v_n).
\end{equation}
The energy and charge functions $H$ and $Q$ are conserved in the time evolution of the Hamiltonian system (\ref{Hamiltonian-PT-dNLS}).
Compared to the other physically relevant $\PT$-symmetric dNLS equations \cite{pel1,pel3,pel4},
where the Hamiltonian structure is not available and analysis of nonlinear stability of the zero
equilibrium and time-periodic localized breathers is barely possible, we are able to address
these questions for the $\PT$-symmetric dNLS equation (\ref{PT-dNLS}), thanks to
the Hamiltonian structure (\ref{Hamiltonian-PT-dNLS}) with two conserved quantities (\ref{energy-PT-dNLS}) and (\ref{charge-PT-dNLS}).

The temporal evolution of the $\PT$-symmetric dNLS equation (\ref{PT-dNLS}) is studied in
sequence space $\ell^2(\mathbb{Z})$ for sequences $(u,v)$ as functions of time. Global existence
of solutions in $\ell^2(\mathbb{Z})$ follows from an easy application of Picard's method and
energy estimates (Proposition \ref{proposition-existence}). The global solution in $\ell^2(\mathbb{Z})$
may still grow at most exponentially in time, due to the destabilizing properties of the gain-damping terms
in the system (\ref{PT-dNLS}). However, thanks to coercivity of the energy function (\ref{energy-PT-dNLS}) near the zero
equilibrium, we can still obtain a global bound on the $\ell^2(\mathbb{Z})$ norm of the solution
near the zero equilibrium, provided it is linearly stable. Moreover, for $\Omega > (\gamma + 4 \epsilon)$,
the global bound holds for arbitrary initial data.
The corresponding result is given by the following theorem (proved in Section 2).

\begin{theorem}
For every $\Omega > (\gamma + 4 \epsilon)$ and every initial data $(u(0),v(0)) \in \ell^2(\mathbb{Z})$,
there is a positive constant $C$ that depends on parameters and $\Omega, \gamma, \epsilon$ and
$(\| u(0) \|_{\ell^2},\| v(0) \|_{\ell^2})$ such that
\begin{equation}
\label{time-ind-bound}
\| u(t) \|_{\ell^2}^2 + \| v(t) \|_{\ell^2}^2 \leq C, \quad \mbox{\rm for every} \;\; t \in \mathbb{R}.
\end{equation}
Furthermore, the bound (\ref{time-ind-bound}) holds for every $\Omega < -\gamma$ and
every $(u(0),v(0)) \in \ell^2(\mathbb{Z})$ with sufficiently small $\ell^2(\mathbb{Z})$ norm.
\label{theorem-bound}
\end{theorem}

\begin{remark}
As shown in \cite{ChernPel}, the zero equilibrium of the $\PT$-symmetric dNLS equation (\ref{PT-dNLS})
is linearly stable if $|\gamma| < \gamma_0$, where the $\PT$ phase transition threshold $\gamma_0$ is
given by
\begin{equation}
\label{gamma-0}
\gamma_0 := \left\{ \begin{array}{ll} \Omega - 4 \epsilon, & \Omega > 0 \\
|\Omega|, & \Omega < 0. \end{array} \right.
\end{equation}
The zero equilibrium is linearly unstable if $|\gamma| \geq \gamma_0$.
Thus, the constraints on parameters in Theorem \ref{theorem-bound}
coincide with the criterion of linear stability of the zero equilibrium.
\end{remark}

We shall now characterize breathers supported by the $\PT$-symmetric dNLS equation (\ref{PT-dNLS}).
These are solutions of the form
\begin{equation}
\label{stationary}
u(t) = U e^{-i E t}, \quad v(t) = V e^{-i E t},
\end{equation}
where the frequency  parameter $E$ is considered to be real and
the sequence $(U,V) \in \ell^2(\mathbb{Z})$ is time-independent. By continuous embedding,
we note that $(U,V) \in \ell^2(\mathbb{Z})$ implies the decay at infinity: $|U_n| + |V_n| \to 0$
as $|n| \to \infty$. The breather is considered to be $\PT$-symmetric
with respect to the operators in (\ref{PT-symmetry}) if $V = \bar{U}$.

Thanks to the cross-gradient symplectic structure (\ref{Hamiltonian-PT-dNLS}),
the breather solution (\ref{stationary}) is a critical point of
the extended energy function $H_E : \ell^2(\mathbb{Z}) \to \mathbb{R}$ given by
\begin{equation}
\label{combined-energy-functional}
H_E := H - E Q,
\end{equation}
where $H$ and $Q$ are given by (\ref{energy-PT-dNLS}) and (\ref{charge-PT-dNLS}).
The Euler--Lagrange equations for $H_E$ produce
the stationary $\PT$-symmetric dNLS equation:
\begin{equation}
E U_n = \epsilon \left( \bar{U}_{n+1} - 2 \bar{U}_n + \bar{U}_{n-1} \right) + i \gamma U_n + \Omega \bar{U}_n  +
6 |U_n|^2 \bar{U}_n + 2 U_n^3, \label{eq:statPT}
\end{equation}
which corresponds to the reduction of the $\PT$-symmetric dNLS equation (\ref{PT-dNLS})
for the breather solution (\ref{stationary}) under the $\PT$ symmetry $V = \bar{U}$.

Existence and spectral stability of breathers can be characterized
in the limit of small coupling constant $\epsilon$, when breathers
bifurcate from solutions of the dimer equation arising at a single site, say the central site at $n = 0$.
This technique was introduced for the $\PT$-symmetric systems in \cite{pel2,pel3}
and was applied to the system of amplitude equations (\ref{PT-dNLS}) in \cite{ChernPel}.
Here we recall the main facts about these breathers obtained in \cite{ChernPel}.

Figure \ref{branches} represents branches of the time-periodic solutions of the central dimer at $\epsilon = 0$,
where the amplitude of the central dimer $A = |U_0| = |V_0|$ is plotted versus the frequency parameter $E$.
The left panel corresponds to the solution with $\Omega > \gamma > 0$, whereas
the right panel corresponds to the solution with $\Omega < - \gamma < 0$.
The constraint $|\gamma| < |\Omega|$ is used for stability of the zero equilibrium at $\epsilon = 0$
outside the central dimer, according to Theorem \ref{theorem-bound}.
The values $\pm E_0$ with $E_0 := \sqrt{\Omega^2-\gamma^2}$ correspond to bifurcation of the small-amplitude
solutions. The small-amplitude solutions are connected with the large-amplitude solutions for $\Omega > \gamma > 0$, whereas
the branches of small-amplitude and large-amplitude solutions are disconnected for $\Omega < -\gamma < 0$.

\begin{figure}[!htbp]
\center
\includegraphics[scale=0.6]{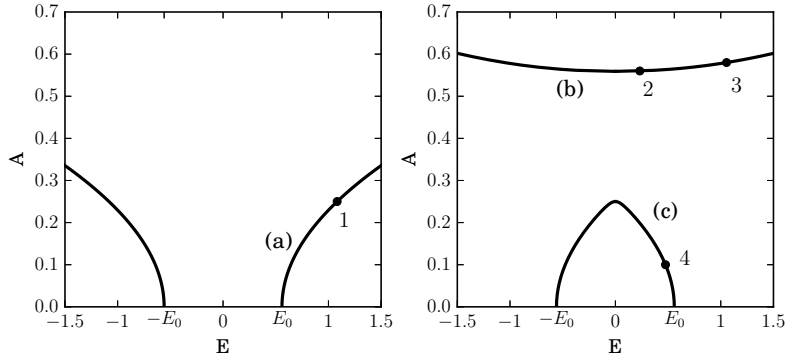}
\caption{Time-periodic solutions of the $\PT$-symmetric dimer for $\gamma = \frac{1}{2}$ and (a) $\Omega = \frac{3}{4} > \gamma$ or
(b) $\Omega = -\frac{3}{4} < -\gamma$.}
\label{branches}
\end{figure}

Every time-periodic solution supported at the central dimer for $\epsilon = 0$
is continued uniquely and smoothly with respect to the small coupling parameter $\epsilon$
by the implicit function arguments \cite{ChernPel}. The resulting breather is symmetric
about the central site and $\PT$-symmetric so that
\begin{equation}
\label{symmetric}
V_n = \bar{U}_n = \bar{U}_{-n} = V_{-n}, \quad n \in \mathbb{Z}.
\end{equation}
Moreover, the breather profile decays fast at infinity (Proposition \ref{proposition-soliton}).

Since $(U,V)$ are critical points of the extended energy function (\ref{combined-energy-functional}),
the nonlinear stability of breathers can be studied by the Lyapunov method
if the second variation of $H_E$ is sign-definite in $\ell^2(\mathbb{Z})$.
The second variation of $H_E$ is given by a quadratic form associated with
the self-adjoint (Hessian) operator $\mathcal{H}''_E : \ell^2(\mathbb{Z}) \to \ell^2(\mathbb{Z})$ written in the form
\begin{eqnarray}
\label{PT-dNLS-var}
\mathcal{H}''_E = \mathcal{M} + \epsilon \mathcal{L},
\end{eqnarray}
where blocks of $\mathcal{M}$ at each lattice site $n \in \mathbb{Z}$ are given by
{\small \begin{eqnarray*}
\arraycolsep=0pt\def\arraystretch{1}
\mathcal{M}_n = \left[\begin{array}{cccc} \Omega + 8 |U_n|^2 & 2(U_n^2 + \bar{U}_n^2) & -E - i \gamma + 4(U_n^2 + \bar{U}_n^2) & 4 |U_n|^2 \\
2(U_n^2 + \bar{U}_n^2) & \Omega + 8 |U_n|^2 & 4 |U_n|^2  & -E + i \gamma + 4(U_n^2 + \bar{U}_n^2) \\
-E + i \gamma + 4(U_n^2 + \bar{U}_n^2) & 4 |U_n|^2 & \Omega + 8 |U_n|^2 & 2(U_n^2 + \bar{U}_n^2)\\
4 |U_n|^2 & -E - i \gamma + 4(U_n^2 + \bar{U}_n^2) & 2(U_n^2 + \bar{U}_n^2) & \Omega + 8 |U_n|^2 \end{array} \right],
\end{eqnarray*}}whereas $\mathcal{L}$ is the discrete Laplacian operator applied to blocks of $\phi$
at each lattice site $n \in \mathbb{Z}$:
$$
(\mathcal{L} \phi)_n = \phi_{n+1} - 2 \phi_n + \phi_{n-1}.
$$

For the two solution branches with $\Omega < -\gamma < 0$ and $|E| < E_0$
(points $2$ and $4$ on Figure \ref{branches}), it was shown in \cite{ChernPel}
that the infinite-dimensional part of the spectrum of $\mathcal{H}''_E$ in $\ell^2(\mathbb{Z})$
is negative definite and the rest of the spectrum includes a simple zero eigenvalue due to gauge symmetry
and either three (in case of point $2$) or one (in case of point $4$) positive eigenvalues. As a result,
the nonlinear orbital stability of the corresponding breathers was developed in \cite{ChernPel}
by using the standard energy methods \cite{Cuccagna,Kapitula}.

On the other hand, for the solution branches with $|E| > E_0$ (points $1$ and $3$ on Figure \ref{branches}),
it was shown in \cite{ChernPel} that the spectrum of $\mathcal{H}''_E$ in $\ell^2(\mathbb{Z})$
includes infinite-dimensional positive and negative parts. Therefore, for $|E| > E_0$ both
for $\Omega > \gamma > 0$ and $\Omega < -\gamma < 0$,
$(U,V)$ is an infinite-dimensional saddle point of the extended energy function $H_E$. This is very similar to the Hamiltonian
systems of the Dirac type, where the zero equilibrium and standing waves are located in the gap
between the positive and negative continuous spectrum.

Spectral stability of the solution branches with $\Omega > \gamma > 0$ and $|E| > E_0$ is proved for
sufficiently small $\epsilon$ under the non-resonance condition, which is checked numerically \cite{ChernPel}.
On the other hand, the solution branch with $\Omega < -\gamma < 0$ and $|E| > E_0$
is spectrally stable for sufficiently small $\epsilon$ almost everywhere except for
the narrow interval in the parameter space, where the non-resonance condition is not satisfied \cite{ChernPel}.
Since in both cases, $(U,V)$ is an infinite-dimensional saddle point
of the extended energy function $H_E$, the standard energy methods \cite{Kapitula} can not be applied
to the proof of nonlinear stability of the solution branches with $|E| > E_0$.

The main contribution of this paper is a proof of long-time nonlinear stability of the
infinite-dimensional saddle point $(U,V)$ by using the asymptotic limit
of small coupling parameter $\epsilon$. The novel method
which we develop here works for the solution branches with $\Omega > \gamma > 0$
and $|E| > E_0$ (shown on the left panel of Figure \ref{branches}). To remedy the difficulty with the
energy method, we select the energy function in the form
\begin{equation}
\Lambda_E :=  H - E (u_0 \bar{v}_0 + \bar{u}_0 v_0).
\label{lyapunov-function}
\end{equation}
Note that $\Lambda_E$ is different from the extended energy function $H_E$ in (\ref{combined-energy-functional}),
since $\Lambda_E$ only includes the part of $Q$ at the central site $n = 0$, where
$(U,V)$ is supported if $\epsilon = 0$.
With the definition of $\Lambda_E$ given by (\ref{lyapunov-function}),
we obtain a function with a positive second variation at $(U,V)$, however, two new obstacles arise now:
\begin{itemize}
\item[(i)] the first variation of $\Lambda_E$ does not vanish at $(U,V)$ if $\epsilon \neq 0$;
\item[(ii)] the value of $\Lambda_E$ is no longer constant in the time evolution
of the dNLS equation (\ref{PT-dNLS}).
\end{itemize}
The first difficulty is overcome with a local transformation of dependent variables.
However, due to the second difficulty, instead of the nonlinear stability for all times,
as in Lyapunov's stability theorem (Appendix A), we only establish
a long-time nonlinear stability of the breather on a long but finite time interval.
This long-time stability is usually referred as {\em metastability}.

We denote a solution of the $\PT$-symmetric dNLS equation (\ref{PT-dNLS}) in $\ell^2(\mathbb{Z})$ by
$\psi = (u,v)$ and the localized solution of the stationary dNLS equation (\ref{eq:statPT}) by
$\Phi = (U,V)$. We fix parameters $\gamma > 0$, $\Omega > \gamma$, and $E \in (-\infty,-E_0) \cup (E_0,\infty)$.
The following theorem  (proved in Section 3) formulates the main result of this paper.

\begin{theorem}
For every $\nu > 0$ sufficiently small, there exists $\epsilon_0 > 0$ and $\delta > 0$ such that
for every $\epsilon \in (0,\epsilon_0)$ the following is true. If $\psi(0) \in \ell^2(\mathbb{Z})$
satisfies $\| \psi(0) - \Phi \|_\ltz \le \delta$, then there exist a positive time $t_0 \lesssim \epsilon^{-1/2}$
and a $C^1$ function $\alpha(t) : [0,t_0] \to \mathbb{R} / (2\pi \mathbb{Z})$ such that
the unique solution $\psi(t) : [0,t_0] \to \ell^2(\mathbb{Z})$ to the $\PT$-symmetric dNLS equation
(\ref{PT-dNLS}) satisfies the bound
\begin{equation}
\label{bound-final}
\| e^{i\alpha(t)} \psi(t)  - \Phi \|_\ltz \le \nu, \quad \mbox{\rm for every} \;\; t \in [0,t_0].
\end{equation}
Moreover, there exists a positive constant $C$ such that $|\dot{\alpha} - E| \leq C \nu$, for every $t \in [0,t_0]$.
\label{theorem-main}
\end{theorem}

\begin{remark}
The statement of Theorem \ref{theorem-main} remains true for $\epsilon = 0$. In this (anti-continuum)
limit, Theorem \ref{theorem-main} gives nonlinear stability
of the standing localized state $\Phi$ compactly supported at the central site $n = 0$. The bound (\ref{bound-final})
is extended in the case $\epsilon = 0$ for all times $t \in \mathbb{R}$.
\end{remark}

\begin{remark}
It becomes clear from the proof of Theorem \ref{theorem-main} for $\epsilon \neq 0$, see inequality (\ref{estimate-parameters}) below,
that the bound (\ref{bound-final}) on the perturbation $\phi$ to the stationary solution $\Phi$ is defined within
the size of $\mathcal{O}(\epsilon^{1/2} + \delta)$. Therefore, if $\Phi_n = \mathcal{O}(\epsilon^{|n|})$ for every $n \neq 0$
(Proposition \ref{proposition-soliton}), then the perturbation term is $\phi_n = \mathcal{O}(\epsilon^{1/2} + \delta)$
for every $n \in \mathbb{Z}$. This is a limitation of the result of Theorem \ref{theorem-main}. Not only it
holds for long but finite times $t_0 = \mathcal{O}(\epsilon^{-1/2})$ but also it gives a larger than expected
bound on the perturbation term $\phi$.
\end{remark}

\begin{remark}
\label{remark-shorter-scale}
The statement of Theorem \ref{theorem-main} can be improved on a shorter time scale $t_0 = \mathcal{O}(1)$.
In this case, see  inequality (\ref{estimate-parameters-short}) below,
the perturbation term $\phi$ has the size of $\mathcal{O}(\epsilon + \delta)$.
Thus, the perturbation term $\phi_n$ at $n = \pm 1$
is comparable with the standing localized state $\Phi_n$ at $n = \pm 1$, but it is still much larger than
$\Phi_n$ for every $n$ such that $|n| \geq 2$.
\end{remark}

\begin{remark}
\label{remark-other-solutions}
Theorem \ref{theorem-main} cannot be extended to the solution branch with $\Omega < -\gamma < 0$ and $|E| > E_0$
(shown on the right panel of Figure \ref{branches})
because the second variation of $\Lambda_E$ at $(U,V)$ is not coercive and does not control
the size of perturbation terms. This analytical difficulty reflects the unfortunate location of
the discrete and continuous spectra that leads to a resonance studied in \cite{ChernPel}.
No resonance was found for the solution branch with $\Omega > \gamma > 0$ and $|E| > E_0$
(shown on the left panel of Figure \ref{branches})
and this numerical result from \cite{ChernPel}
is in agreement with the analytical method used in the proof of Theorem \ref{theorem-main}.
\end{remark}

The remainder of this paper is devoted to the proof of Theorems \ref{theorem-bound} and \ref{theorem-main}.

\section{Proof of Theorem \ref{theorem-bound}}

The following proposition gives the global existence result
for the $\PT$-symmetric dNLS equation (\ref{PT-dNLS}).

\begin{proposition}
\label{proposition-existence}
For every $(u^{(0)},v^{(0)}) \in \ell^2(\mathbb{Z})$, there exists a unique solution $(u,v)(t) \in
C^1(\mathbb{R},\ell^2(\mathbb{Z}))$ of the $\PT$-symmetric dNLS equation (\ref{PT-dNLS}) such that
$(u,v)(0) = (u^{(0)},v^{(0)})$. The unique solution depends continuously on initial data
$(u^{(0)},v^{(0)}) \in \ell^2(\mathbb{Z})$.
\end{proposition}

\begin{proof}
Since discrete Laplacian is a bounded operator in $\ell^2(\mathbb{Z})$ and the sequence space
$\ell^2(\mathbb{Z})$ forms a Banach algebra with respect to pointwise multiplication,
the local well-posedness of the initial-value problem for the $\PT$-symmetric dNLS equation (\ref{PT-dNLS})
follows from the standard Picard's method.  The local solution $(u,v)(t)$ exists in $C^0([-t_0,t_0],\ell^2(\mathbb{Z}))$
for some finite $t_0 > 0$. Thanks again to the boundedness of the discrete Laplacian operator in $\ell^2(\mathbb{Z})$,
bootstrap arguments extend this solution in $C^1([-t_0,t_0],\ell^2(\mathbb{Z}))$.

The local solution is continued globally by using the energy method.
For any solution $(u,v)(t)$ in $C^1([-t_0,t_0],\ell^2(\mathbb{Z}))$, we obtain
the following balance equation from system (\ref{PT-dNLS}):
$$
\frac{d}{dt} \sum_{n \in \mathbb{Z}} (|u_n|^2 + |v_n|^2) = - \gamma \sum_{n \in \mathbb{Z}} (|u_n|^2 - |v_n|^2).
$$
Integrating this equation in time and applying Gronwall's inequality, we get
$$
\| u(t) \|_{l^2}^2 + \| v(t) \|_{l^2}^2 \le \left(\| u(0)\|_{l^2}^2 + \|v(0) \|_{l^2}^2\right)
e^{|\gamma t|}, \quad t \in [-t_0,t_0].
$$
Therefore $\| u(t)\|_{l^2}$ and $\| v(t)\|_{l^2}$ cannot blow up in a finite time,
so that the local solution $(u,v)(t) \in C^1([-t_0,t_0],\ell^2(\mathbb{Z}))$
is continued for every $t_0 > 0$.
\end{proof}

A critical question also addressed in \cite{pel1,pel3,pel4} for other $\PT$-symmetric dNLS equations
is whether the $\ell^2(\mathbb{Z})$ norms of the global solution of Proposition \ref{proposition-existence}
remain bounded as $t \to \infty$. In the context of the Hamiltonian $\PT$-symmetric dNLS equation (\ref{PT-dNLS}),
this question can be addressed by using the energy function given by (\ref{energy-PT-dNLS}).
In what follows, we use coercivity of the energy function and
prove the result stated in Theorem \ref{theorem-bound}. \\

{\em Proof of Theorem \ref{theorem-bound}.} We recall that $\gamma > 0$ and $\epsilon > 0$ everywhere in the proof.
If $\Omega > (\gamma + 4 \epsilon)$, the following lower bound is available for the energy function $H$ given by (\ref{energy-PT-dNLS})
using Cauchy--Schwarz inequality:
\begin{equation}
\label{coercivity-E}
H \ge (\Omega - \gamma - 4 \epsilon) \left(\| u \|_{\ell^2}^2 + \| v \|_{\ell^2}^2\right).
\end{equation}
Since $H$ is time-independent and bounded for any $(u,v)(t) \in C^1(\mathbb{R},\ell^2(\mathbb{Z}))$
due to the continuous embedding $\| u \|_{\ell^p} \leq \| u \|_{\ell^2}$ for any $p \geq 2$,
we obtain the time-independent bound (\ref{time-ind-bound}) for any $\Omega > (\gamma + 4 \epsilon)$.

If $\Omega < -\gamma$, the following lower bound is available for the energy function $-H$:
\begin{equation}
\label{coercivity-E-pos}
-H \ge (|\Omega| - \gamma) \left(\| u \|_{\ell^2}^2 + \| v \|_{\ell^2}^2\right) - \left(\| u \|_{\ell^2}^2 + \| v \|_{\ell^2}^2\right)^2,
\end{equation}
where the continuous embedding $\| u \|_{\ell^4} \leq \| u \|_{\ell^2}$ has been used.
If $\| u(0) \|_{\ell^2} + \| v(0) \|_{\ell^2}$ is sufficiently small, then $|H|$ is sufficiently small,
and the bound (\ref{time-ind-bound}) with sufficiently small $C$ holds for every $t \in \mathbb{R}$. $\Box$ \\

\begin{remark}
For every $\Omega < (\gamma + 4 \epsilon)$, the energy functions $H$ or $-H$
do not produce a useful lower bound, which would
result in a time-independent bound on the $\ell^2(\mathbb{Z})$ norm for the global solution
$(u,v)(t)$. This is because the continuous embedding $\| u \|_{\ell^4} \leq \| u \|_{\ell^2}$
is not sufficient to control $H$ or $-H$ from below.

If the lattice is truncated on a finitely many (say, $N$) sites, then
the bound $\| u \|_{\ell^2} \leq N^{1/4} \| u \|_{\ell^4}$ can be used to obtain from (\ref{energy-PT-dNLS}):
$$
H \ge \left(\| u \|_{\ell^4}^4 + \| v \|_{\ell^4}^4\right) - (\gamma + 4 \epsilon - \Omega) N^{1/2}
\left(\| u \|_{\ell^4}^2 + \| v \|_{\ell^4}^2\right).
$$
Thus, the time-independent bound on the $\ell^4(\mathbb{Z}_N)$ (and then $\ell^2(\mathbb{Z}_N)$)
norms for the global solution $(u,v)(t)$ restricted on $N$ sites of the lattice $\mathbb{Z}$ is available
for every $\Omega$. However, the control becomes impossible in the limit $N \to \infty$ if
$\Omega < (\gamma + 4 \epsilon)$.
\end{remark}

It is an interesting open question to investigate if the global dynamics of
the $\PT$-symmetric dNLS equation (\ref{PT-dNLS}) on the infinite lattice
is globally bounded in time for $\Omega \leq (\gamma + 4 \epsilon)$.
This open question would include the case $-\gamma \leq \Omega \leq (\gamma + 4 \epsilon)$,
when the zero equilibrium is linearly unstable, and the
case $\Omega < -\gamma$ with sufficiently large initial data $(u(0),v(0)) \in \ell^2(\mathbb{Z})$,
when the zero equilibrium  is linearly stable but the bound (\ref{coercivity-E-pos}) can no longer be closed.

\section{Proof of Theorem \ref{theorem-main}}

We divide the proof of Theorem \ref{theorem-main} into several subsections.

\subsection{Characterization of the localized solutions}

For $\epsilon = 0$, a solution to the stationary dNLS equation (\ref{eq:statPT})
is supported on the central site $n = 0$ and satisfies
\begin{equation}
\label{dimer-states}
(E - i \gamma) U_0 - \Omega \bar{U}_0 = 6 |U_0|^2 \bar{U}_0 + 2 U_0^3.
\end{equation}
The parameters $\gamma$ and $\Omega$ are considered to be fixed, and parameter $E$ is thought
to parameterize a continuous branch of solutions of the nonlinear
algebraic equation (\ref{dimer-states}). Substituting the decomposition $U_0 = A e^{i \theta}$
with $A > 0$ and $\theta \in [0,2\pi)$ into the algebraic equation (\ref{dimer-states}), we obtain
\begin{equation}
\label{parameterization-A}
\sin(2\theta) = \frac{\gamma}{4 A^2 + \Omega}, \quad \cos(2 \theta) = \frac{E}{8 A^2 + \Omega},
\end{equation}
from which the solution branches of $E$ versus $A$ are obtained in  \cite{ChernPel}
as shown on Figure \ref{branches}. The dependence of $E$ versus $A$ is given analytically by
\begin{equation}
\label{branch-E-A}
E^2 = (\Omega + 8 A^2)^2 \left[ 1 - \frac{\gamma^2}{(\Omega + 4 A^2)^2} \right].
\end{equation}
Persistence of the central dimer in the unbounded lattice
with respect to the coupling parameter  $\epsilon$ is given by  the following proposition.

\begin{proposition}
\label{proposition-soliton}
Fix $\gamma > 0$, $\Omega > \gamma$, and $E \neq \pm E_0$,
where $E_0 := \sqrt{\Omega^2 - \gamma^2} > 0$.
There exist $\epsilon_0 > 0$ sufficiently small and $C_0 > 0$ such that
for every $\epsilon \in (-\epsilon_0,\epsilon_0)$, there exists a unique
solution $U \in l^2(\mathbb{Z})$ to the difference equation (\ref{eq:statPT}) such that
\begin{equation}
\label{bound-soliton}
\left| U_0 - A e^{i \theta} \right| \leq C_0 |\epsilon|, \quad |U_n| \leq C_0 |\epsilon|^{|n|}, \quad n \neq 0,
\end{equation}
where $A$ and $\theta$ are defined in (\ref{parameterization-A}). Moreover, the solution $U$
is smooth in $\epsilon$.
\end{proposition}

\begin{proof}
Persistence and smoothness of a solution $U \in l^2(\mathbb{Z})$ to the difference equation (\ref{eq:statPT})
in $\epsilon$ is proved with two applications of the implicit function theorem
(Appendix A). In the first application, we consider the following system of algebraic equations
\begin{equation}
E U_n = \epsilon \left( \bar{U}_{n+1} - 2 \bar{U}_n + \bar{U}_{n-1} \right) + i \gamma U_n + \Omega \bar{U}_n +
6 |U_n|^2 \bar{U}_n + 2 U_n^3, \quad \pm n \in \mathbb{N}, \label{fixed-point-1}
\end{equation}
where $U_0 \in \mathbb{C}$ is considered to be given, in addition to the given parameters $\gamma$, $\Omega$, and $E$.
Let $x = \{ U_n \}_{\pm n \in \mathbb{N}}$, $X = \ell^2(\mathbb{N})$, $y = \epsilon$, $Y = \mathbb{R}$,
and $Z = \ell^2(\mathbb{N})$ in the definition of system $F : X \times Y \to Z$.
Then we have $F(0,0) = 0$ and the Jacobian
operator $D_x F(0,0) : X \to Z$ is block-diagonal with identical blocks given by
\begin{equation}
\label{jacobian-DxF}
\left[ \begin{matrix} E - i \gamma & -\Omega \\ -\Omega & E + i \gamma \end{matrix} \right],
\end{equation}
with the eigenvalues $\lambda_{\pm} := E \pm E_0$, where $E_0 := \sqrt{\Omega^2 - \gamma^2}$.
By the assumption of the proposition,
$\lambda_{\pm} \neq 0$, so that the Jacobian operator $D_x F(0,0)$ is one-to-one and onto.
By the implicit function theorem, for every $\epsilon \neq 0$ sufficiently small, there exists a unique small
solution $U \in \ell^2(\mathbb{N})$ to the system (\ref{fixed-point-1}) such that
\begin{equation}
\| U \|_{l^2(\mathbb{N})} \leq C' |\epsilon| |U_0|,
\label{eq:U}
\end{equation}
where a positive constant $C'$ is independent from $\epsilon$.
By the symmetry of the two systems (\ref{fixed-point-1}) for $\pm n \in \mathbb{N}$
and uniqueness of solutions, we have $U_{-n} = U_n$ for every $n \in \mathbb{N}$.

In the second application of the implicit function theorem, we consider the
following algebraic equation
\begin{equation}
E U_0 = 2 \epsilon \left( \bar{U}_{1} - \bar{U}_0 \right) + i \gamma U_0 + \Omega \bar{U}_0 +
6 |U_0|^2 \bar{U}_0 + 2 U_0^3, \label{fixed-point-2}
\end{equation}
where $U_1 \in \mathbb{C}$ depends on $U_0$, $\gamma$, $\Omega$, and $E$,
according to the previous result. Let $x = U_0$, $X = \mathbb{C}$, $y = \epsilon$, $Y = \mathbb{R}$,
and $Z = \mathbb{C}$ in the definition of system $F : X \times Y \to Z$.
Then we have $F(A e^{i\theta},0) = 0$, where $A$ and $\theta$ are
defined by (\ref{parameterization-A}). The Jacobian
operator $D_x F(A e^{i \theta},0) : X \to Z$ is given by the matrix
\begin{eqnarray}
\label{matrix-Jacobian}
\left[ \begin{matrix} E - i \gamma - 6 U_0^2 - 6 \bar{U}_0^2 & -\Omega - 12 |U_0|^2 \\ -\Omega - 12 |U_0|^2 &
E + i \gamma - 6 U_0^2 - 6 \bar{U}_0^2 \end{matrix} \right] \biggr|_{U_0 = A e^{i \theta}}.
\end{eqnarray}
It is shown in \cite{ChernPel} that the matrix given by (\ref{matrix-Jacobian}) is invertible
for every $\Omega > \gamma > 0$ and $|E| > E_0$.
By the implicit function theorem, for every $\epsilon \neq 0$ sufficiently small,
there exists a unique solution $U_0 \in \mathbb{C}$ to the algebraic equation (\ref{fixed-point-2})
near $A e^{i \theta}$ such that
\begin{equation}
\left| U_0 - A e^{i \theta} \right| \leq C'' |\epsilon|,
\label{eq:U0}
\end{equation}
where a positive constant $C''$ is independent from $\epsilon$. Since both equations (\ref{fixed-point-1}) and (\ref{fixed-point-2}) are smooth in $\epsilon$,
the solution $U$ is smooth in $\epsilon$. Thus, persistence of a smooth solution $U \in \ell^2(\mathbb{Z})$
with the first bound in (\ref{bound-soliton}) is proved by the implicit function theorem and the bound (\ref{eq:U0}).

It remains to prove the second bound in (\ref{bound-soliton}), for which we employ the
implicit function theorem for the third time. Inspecting the difference equation~\eqref{eq:statPT} shows that if
$U_{\pm 1} = \mathcal{O}(|\epsilon|)$ according to the bound (\ref{eq:U}),
then $U_n$ can be expressed by using the scaling transformation
\begin{equation}
U_n = \epsilon^{|n|} W_{|n|}, \quad \pm n \in \mathbb{N},
\label{eq:solutionse-precise}
\end{equation}
where the sequence $W \in \ell^2(\mathbb{N})$ is found from the system
\begin{equation}
EW_n - i\gamma W_n - \Omega \bar{W}_n = \bar{W}_{n-1} + \epsilon^2 \bar{W}_{n+1}
 - 2\epsilon \bar{W}_n + 6\epsilon^{2 |n|} |W_n|^2 \bar{W}_n + 2\epsilon^{2|n|}W_n^3,
 \quad n \in \mathbb{N}, \label{eq:Wn}
\end{equation}
with $W_0 = U_0$ given by the previous result.
Let $x = \{ W_n \}_{n \in \mathbb{N}}$, $X = \ell^2(\mathbb{N})$,
$y=\epsilon$, $Y = \mathbb{R}$, and $Z = \ell^2(\mathbb{N})$ in the definition
of system $F : X \times Y \to Z$.
Then, we have $F(x_0,0) = 0$, where $x_0 = \{ W_n^{(0)} \}_{n \in \mathbb{N}}$ is a unique solution
of the recurrence equation
\begin{equation}
E W_n^{(0)} - i\gamma W_n^{(0)} - \Omega \bar{W}_n^{(0)} = \bar{W}_{n-1}^{(0)}, \quad n \in \mathbb{N},
\label{leading-order}
\end{equation}
starting with a given $W_0^{(0)} = U_0$. Indeed, each block of (\ref{leading-order}) is given by
the invertible matrix ~\eqref{jacobian-DxF} with eigenvalues $\lambda_{\pm} = E\pm E_0 \neq 0$, hence,
a unique solution for $W^{(0)} \in \ell^{\infty}(\mathbb{N})$ is found from the recurrence relation (\ref{leading-order}).
Moreover, since $D_x F(0,0) : X \to Z$ is one-to-one and onto (as a lower block-triangular matrix with invertible diagonal blocks),
the solution $W^{(0)}$ is actually in $X = \ell^2(\mathbb{N})$. By the implicit function theorem
(Appendix A), for every $\epsilon \neq 0$ sufficiently small,
there exists a unique solution $W \in \ell^2(\mathbb{N})$ to the system (\ref{eq:Wn}) such that
\begin{equation}
\label{bound-W-n}
\| W - W^{(0)} \|_{\ell^2(\mathbb{N})} < C'''|\epsilon|,
\end{equation}
where a positive constant $C'''$ is independent of $\epsilon$. Thus, the
second bound in (\ref{bound-soliton}) is proved from (\ref{eq:solutionse-precise})
and (\ref{bound-W-n}).
\end{proof}

\subsection{Decomposition of the solution}

Let $\psi = (u,\bar{u},v,\bar{v})$ denote a solution of the $\PT$-symmetric dNLS equation (\ref{PT-dNLS}) in $\ell^2(\mathbb{Z})$
given by Proposition \ref{proposition-existence}. Let $\Phi = (U,\bar{U},V,\bar{V})$ denote a localized solution
of the stationary dNLS equation (\ref{eq:statPT}) given by Proposition \ref{proposition-soliton}.
Let $\phi = \psi - \Phi = ({\bf u},\bar{\bf u},{\bf v},\bar{\bf v})$ denote a perturbation to $\Phi$.
Note that these are extended $4$-component variables at each lattice site (concatenated by the complex conjugate
functions) compared to the two-component variables used in the formulation of Theorem \ref{theorem-main}.
The extended variables are more suitable for dealing with the energy functions 
such as (\ref{combined-energy-functional}) or (\ref{lyapunov-function}). 

By using the energy function (\ref{lyapunov-function}), we introduce the energy difference function
\begin{equation}
\Delta := \Lambda_E(\Phi + \phi) - \Lambda_E (\Phi).
\label{eq:convex}
\end{equation}
Let us write the expansion for $\Delta$ explicitly:
\begin{equation}
\label{eq:rearr}
\Delta = N_1(\phi) + N_2(\phi) + N_3(\phi) + N_4(\phi),
\end{equation}
where the linear part is
\begin{equation}
\label{linear}
N_1(\phi) = E \sum_{n \in \mathbb{Z} \backslash \{0\}} \left( \bar{V}_n {\bf u}_n + V_n \bar{\bf u}_n + \bar{U}_n {\bf v}_n + U_n \bar{\bf v}_n \right),
\end{equation}
the quadratic part is
\begin{equation}
\label{quadratic}
N_2(\phi) = \frac{1}{2} \langle \mathcal{H}''_E \phi, \phi \rangle_\ltz
+ E \sum_{n \in \mathbb{Z} \backslash \{0\}} \left( \bar{\bf v}_n {\bf u}_n + {\bf v}_n \bar{\bf u}_n \right),
\end{equation}
whereas the cubic and quartic parts of $\Delta$ denoted by $N_3(\phi)$ and $N_4(\phi)$ are not important
for estimates, thanks to the bounds
\begin{equation}
\label{cubic-quartic}
| N_3(\phi) | \leq C_3 \|\phi\|^3_\ltz, \quad
| N_4(\phi) | \leq C_4 \|\phi\|^4_\ltz,
\end{equation}
where $C_3$, $C_4$ are positive constants and we have used continuous embedding
$\| u \|_{\ell^p} \leq \| u \|_{\ell^2}$ for any $p \geq 2$.

In the next three subsections, we show that the quadratic part $N_2(\phi)$ is positive,
the linear part $N_1(\phi)$ can be removed by a local transformation,
and the time evolution of $\Delta$ can be controlled on a long but finite time interval.

In what follows, all constants depend on parameters $\gamma > 0$, $\Omega > \gamma$, and
$E \in (-\infty,-E_0) \cup (E_0,\infty)$. The parameter $\epsilon > 0$ is sufficiently
small, and unless it is stated otherwise, the constants do not depend on the small parameter $\epsilon$.

\subsection{Positivity of the quadratic part of $\Delta$}

The quadratic part (\ref{quadratic}) can be analyzed by a parameter continuation from the case $\epsilon = 0$.
Compared to the self-adjoint (Hessian) operator $\mathcal{H}''_E : \ell^2(\mathbb{Z}) \to \ell^2(\mathbb{Z})$ given by
(\ref{PT-dNLS-var}), the Hessian operator for $N_2(\phi)$
denoted by $\Lambda''_E : \ell^2(\mathbb{Z}) \to \ell^2(\mathbb{Z})$ is given by
\begin{eqnarray}
\label{PT-dNLS-var-mod}
\Lambda''_E = \mathcal{\tilde{M}} + \epsilon \mathcal{L},
\end{eqnarray}
where the discrete Laplacian $\mathcal{L}$ is the same but the blocks of $\mathcal{\tilde{M}}$
at each site $n \in \mathbb{Z}$ are now given differently for $n = 0$ and $n \neq 0$.
For $n = 0$, $\mathcal{\tilde{M}}_0 = \mathcal{M}_0$, whereas for $n \neq 0$, we have
{\small \begin{eqnarray*}
\mathcal{\tilde{M}}_n = \left[\begin{array}{cccc} \Omega + 8 |U_n|^2 & 2(U_n^2 + \bar{U}_n^2) & - i \gamma + 4(U_n^2 + \bar{U}_n^2) & 4 |U_n|^2 \\
2(U_n^2 + \bar{U}_n^2) & \Omega + 8 |U_n|^2 & 4 |U_n|^2  & + i \gamma + 4(U_n^2 + \bar{U}_n^2) \\
+ i \gamma + 4(U_n^2 + \bar{U}_n^2) & 4 |U_n|^2 & \Omega + 8 |U_n|^2 & 2(U_n^2 + \bar{U}_n^2)\\
4 |U_n|^2 & - i \gamma + 4(U_n^2 + \bar{U}_n^2) & 2(U_n^2 + \bar{U}_n^2) & \Omega + 8 |U_n|^2 \end{array} \right],
\end{eqnarray*}}that is, parameter $E$ is removed from $\mathcal{M}_n$.

The following proposition characterizes eigenvalues of $\mathcal{\tilde{M}}$ at $\epsilon = 0$.

\begin{proposition}
\label{proposition-critical-point}
Fix $\epsilon = 0$, $\gamma > 0$, $\Omega > \gamma$, and $E \neq \pm E_0$,
where $E_0 := \sqrt{\Omega^2 - \gamma^2} > 0$. The matrix block of $\mathcal{\tilde{M}}_n$
has three positive and one zero eigenvalues for $n = 0$ and two double positive eigenvalues for every $n \neq 0$.
\end{proposition}

\begin{proof}
If $\epsilon = 0$, the stationary state of Proposition \ref{proposition-soliton}
is given by $U_n = 0$ for every $n \neq 0$ and $U_0 = A e^{i \theta}$,
where $A$ and $\theta$ are defined by the parametrization (\ref{parameterization-A}).

For $n = 0$, the $4$-by-$4$ matrix block of $\mathcal{\tilde{M}}_n$ is given by
\begin{equation*} \label{block-L-0}
\arraycolsep=-1pt\def\arraystretch{1.5}
\mathcal{\tilde{M}}_0 = \left[\scalebox{0.8}{$\begin{array}{cccc} \Omega + 8 A^2 & 4 A^2 \cos(2 \theta) & -E - i \gamma + 8 A^2 \cos(2 \theta) & 4 A^2 \\
4 A^2 \cos(2 \theta) & \Omega + 8 A^2 & 4 A^2  & -E + i \gamma + 8 A^2 \cos(2 \theta) \\
-E + i \gamma + 8 A^2 \cos(2 \theta) & 4 A^2 & \Omega + 8 A^2 & 4 A^2 \cos(2 \theta) \\
4 A^2 & -E - i \gamma + 8 A^2 \cos(2 \theta) & 4 A^2 \cos(2 \theta) & \Omega + 8 A^2 \end{array}$} \right].
\end{equation*}
Using relations (\ref{parameterization-A}) and (\ref{branch-E-A}), as well as
symbolic computations with MAPLE, we found that the $4$-by-$4$ matrix block $\mathcal{\tilde{M}}_0$ has
a simple zero eigenvalue and three nonzero eigenvalues $\mu_1$, $\mu_2$, and $\mu_3$ given by
\begin{eqnarray}
\label{eigenvalue-1}
\mu_1 & = & 2 (\Omega + 4A^2), \\
\label{eigenvalue-2-3}
\mu_{2,3} & = & \Omega + 12 A^2 \pm \sqrt{(\Omega - 4A^2)^2 + \frac{16 \Omega A^2 \gamma^2}{(\Omega + 4A^2)^2}}.
\end{eqnarray}
It is shown in \cite{ChernPel} that $\mu_1, \mu_2, \mu_3 > 0$ for every point on the solution branch
with $\Omega > \gamma > 0$, and $|E| > E_0$.

For every $n \in \mathbb{Z} \backslash \{0\}$, the $4$-by-$4$ matrix block of $\mathcal{\tilde{M}}_n$ is given by
\begin{equation*}
\label{block-L-n}
\mathcal{\tilde{M}}_n = \left[\begin{array}{cccc} \Omega & 0 & - i \gamma  & 0 \\
0 & \Omega & 0  & + i \gamma  \\
+ i \gamma  & 0 & \Omega  & 0 \\
0 & - i \gamma & 0 & \Omega \end{array} \right].
\end{equation*}
Each block has two double eigenvalues $\mu_+$ and $\mu_-$ given by
$$
\mu_+ = \Omega + \gamma, \quad \mu_- = \Omega - \gamma,
$$
which are positive since $\Omega > \gamma$.
\end{proof}

By Proposition \ref{proposition-critical-point},
if $\epsilon = 0$, then $N_2(\phi) \geq 0$ for every $\phi \in \ell^2(\mathbb{Z})$
and, moreover, $N_2(\phi) = 0$ if and only if $\phi$ is proportional to an eigenvector supported
at $n = 0$. The existence of the zero eigenvalue at $\epsilon = 0$
is related to the gauge symmetry of the $\PT$-symmetric dNLS equation (\ref{PT-dNLS}).
Both for $\epsilon = 0$ and $\epsilon \neq 0$, there exists a nontrivial kernel
of the Hessian operator $\mathcal{H}''_E : \ell^2(\mathbb{Z}) \to \ell^2(\mathbb{Z})$
associated with the standing localized state $(U,V)$, thanks to the identity
\begin{equation}
\label{generalized-kernel}
\mathcal{H}''_E (\sigma \Phi) = 0,
\end{equation}
where the blocks of the eigenvector $\sigma \Phi$ are given by
\begin{equation}
\label{kernel}
(\sigma \Phi)_n := (U_n, - \bar{U}_n, V_n, -\bar{V}_n), \quad n \in \mathbb{Z}.
\end{equation}
In the limit of $\epsilon \to 0$, the eigenvector $\sigma \Phi$ is supported at
the central site $n = 0$ and it corresponds to the zero eigenvalue of
the matrix block $\mathcal{\tilde{M}}_0 = \mathcal{M}_0$. By
using Proposition \ref{proposition-critical-point} and identity (\ref{generalized-kernel}),
we can now state that if $\epsilon = 0$,
then $N_2(\phi) = 0$ if and only if $\phi \in {\rm span}\{\sigma \Phi\}$.

By the perturbation theory for linear operators (Appendix A), the strictly positive part
of $\Lambda''_E$ remains strictly positive for a sufficiently small $\epsilon$. On the other hand,
the simple zero eigenvalue may drift away from zero if $\epsilon \neq 0$.

In order to avoid a problem of degeneracy (or even slight negativity) of
$\Lambda''_E$, we introduce a constrained subspace of $\ell^2(\mathbb{Z})$ by
\begin{equation}
\ltzc = \{ \phi \in\ltzz \colon \quad \langle \sigma \Phi, \phi \rangle_\ltz = 0 \}.
\label{cond:l2c}
\end{equation}
If $\epsilon = 0$ and $\phi$ belongs to $\ltzc$, then the quadratic form $N_2(\phi)$ in (\ref{quadratic}) is strictly positive
and coercive. By the perturbation theory for linear operators (Appendix A), for $\epsilon \neq 0$ sufficiently small,
the quadratic part $N_2(\phi)$ given by (\ref{quadratic}) for $\phi \in \ltzc$,
remains strictly positive and coercive. This argument yields the proof of the following proposition.

\begin{proposition}
\label{proposition-coercivity}
Fix $\gamma > 0$, $\Omega > \gamma$, and $E \neq \pm E_0$. There exist $\epsilon_0 > 0$ sufficiently small
and $C_2 > 0$  such that for every $\epsilon \in (-\epsilon_0,\epsilon_0)$,
\begin{equation}
\label{coercivity}
N_2(\phi) \geq C_2 \| \phi \|_{\ell^2}^2 \quad \mbox{\rm for every} \;\; \phi \in \ltzc.
\end{equation}
\end{proposition}

Bounds (\ref{cubic-quartic}) and (\ref{coercivity}) allow us to estimate the principal part
of $\Delta$ in (\ref{eq:rearr}) from below, e.g.
$$
| \Delta - N_1(\phi) | \geq \left( C_2 - C_3 \| \phi \|_{\ell^2} - C_4 \| \phi \|_{\ell^2}^2 \right) \| \phi \|_{\ell^2}^2
\quad \mbox{\rm for every} \;\; \phi \in \ltzc.
$$
However, the linear part $N_1(\phi)$ is an obstacle for such estimates.
Therefore, we need to remove the obstacle by a local transformation.

\subsection{Removal of the linear part of $\Delta$}

Let us define
\begin{equation}
\label{near-identity-rho}
\phi = \tilde{\phi} + \rho,
\end{equation}
where $\tilde{\phi} = (\tilde{\bf u}_n,\overline{\tilde{\bf u}_n},\tilde{\bf v}_n,\overline{\tilde{\bf v}}_n)$
is a new variable and $\rho = (a,\bar{a},b,\bar{b})$ is a correction term to be found uniquely by removing the linear term
$N_1(\phi)$. Since the breather is $\PT$-symmetric with $V = \bar{U}$, we shall look for a $\PT$-symmetric
correction term with $b = \bar{a}$.

The easiest way of finding $a \in \ell^2(\mathbb{Z})$ is to write the Euler--Lagrange equations
for the energy function $\Lambda_E$ given by (\ref{lyapunov-function}). For the $\PT$-symmetric
solution with $v = \bar{u}$, the Euler--Lagrange equations for $\Lambda_E$ take the form
\begin{equation}
E u_n \delta_{n,0} = \epsilon \left( \bar{u}_{n+1} - 2 \bar{u}_n + \bar{u}_{n-1} \right) + i \gamma u_n + \Omega \bar{u}_n  +
6 |u_n|^2 \bar{u}_n + 2 u_n^3, \label{Euler-Lagrange}
\end{equation}
where $\delta_{n,0}$ is the Kronecker symbol supported at $n = 0$. Let $u = U + a$, where $U$ is a solution
of the stationary dNLS equation (\ref{eq:statPT}). Then, $a$ satisfies the nonlinear equation
\begin{eqnarray}
\nonumber
E a_n \delta_{n,0} - \Omega \bar{a}_n - i\gamma a_n
-\epsilon (\bar{a}_{n+1} - 2 \bar{a}_n + \bar{a}_{n-1})
- 12 |U_n|^2 \bar{a}_n \quad\quad &\\
- 6(U_n^2 + \bar{U}_n^2) a_n
- 6 U_n (a_n^2 + \bar{a}_n^2) - 12 \bar{U}_n |a_n|^2
- 6 |a_n|^2 \bar{a}_n - 2 a_n^3 & \!\!\! = E U_n (1-\delta_{n,0}),
\label{PT-near-identity}
\end{eqnarray}
where $n \in \mathbb{Z}$. Thanks to the expansion (\ref{eq:solutionse-precise}) in
Proposition \ref{proposition-soliton}, the right-hand side of system (\ref{PT-near-identity}) is small in $\epsilon$.
The following proposition characterizes a unique solution to system (\ref{PT-near-identity}).
This solution with $b = \bar{a}$ defines a unique $\rho$ in the transformation (\ref{near-identity-rho}).

\begin{proposition}
\label{proposition-near-identity}
Fix $\gamma > 0$, $\Omega > \gamma$, and $E \neq \pm E_0$.
There exist $\epsilon_0 > 0$ sufficiently small and $C_1 > 0$ such that
for every $\epsilon \in (-\epsilon_0,\epsilon_0)$, there exists a unique solution
$a \in \ell^2(\mathbb{Z})$ to the system (\ref{PT-near-identity}) such that
\begin{equation}
|a_0| \leq C_1 \epsilon^2, \quad |a_n| \leq C_1 |\epsilon|^{|n|}, \quad n \in \mathbb{Z} \backslash \{0\}.
\label{near-identity-solution}
\end{equation}
\end{proposition}

\begin{proof}
The proof repeats the three steps in the proof of Proposition \ref{proposition-soliton}.
On the sites $n \in \mathbb{Z} \backslash \{0\}$, the Jacobian operator $D_x F(0,0)$ is block-diagonal
with identical blocks given by
\begin{equation}
\label{PT-rho-block-n}
\left[\begin{array}{cc}
-i\gamma & -\Omega \\
-\Omega & i \gamma
\end{array}\right].
\end{equation}
Each block is invertible thanks to the constraint $\Omega > \gamma$. On
the central site $n = 0$, the Jacobian operator $D_x F(A e^{i\theta},0)$ coincides with the block (\ref{matrix-Jacobian}),
which is invertible for every $\gamma \neq 0$, $\Omega > \gamma > 0$, and $|E| > E_0$ \cite{ChernPel}.
Thus, existence and uniqueness of solutions to the nonlinear system
~\eqref{PT-near-identity} for small $\epsilon$ is established with two applications
of the implicit function theorem.

In order to justify the bound (\ref{near-identity-solution}), we use (\ref{eq:solutionse-precise})
and substitute
\begin{equation}
a_0 = \epsilon^2 A_0, \quad a_n = \epsilon^{|n|} A_{|n|}, \quad \pm n \in \mathbb{N}
\label{near-identity-expansion}
\end{equation}
to the system (\ref{PT-near-identity}). The sequence $\{ A_n \}_{n \in \mathbb{N}}$ is found from the system
\begin{eqnarray}
\nonumber
-\Omega \bar{A}_n - i\gamma A_n -
\epsilon^2 \bar{A}_{n+1} + 2 \epsilon \bar{A}_n
- \bar{A}_{n-1} (1 - \delta_{n,1}) - \epsilon^2 A_0 \delta_{n,1} \quad & \\
\nonumber
- 6 \epsilon^{2|n|} (W_n^2 + \bar{W}_n^2) A_n
- 12 \epsilon^{2|n|} |W_n|^2 \bar{A}_n
- 6 \epsilon^{2|n|} W_n (A_n^2 + \bar{A}_n^2) & \\
- 12 \epsilon^{2|n|} \bar{W}_n |A_n|^2
- 6 \epsilon^{2|n|} |A_n|^2 \bar{A}_n - 2 \epsilon^{2|n|} A_n^3 &\!\!\! = E W_n,
\label{near-identity-leading-order}
\end{eqnarray}
whereas the term $A_0$ satisfies the nonlinear equation
\begin{eqnarray}
\nonumber
E A_0 - \Omega \bar{A}_0 - i\gamma A_0 -
2 \bar{A}_{1} + 2 \epsilon \bar{A}_0 - 6(U_0^2 + \bar{U}_0^2) A_0 - 12 |U_0|^2 \bar{A}_0  & \\
- 6 \epsilon^2 U_0 (A_0^2 + \bar{A}_0^2) - 12 \epsilon^2 \bar{U}_0 |A_0|^2 - 6 \epsilon^4 |A_0|^2 \bar{A}_0
- 2 \epsilon^2 A_0^3 & \!\!\! = 0.
\label{near-identity-leading-order-zero}
\end{eqnarray}
It follows from the invertibility of the block (\ref{matrix-Jacobian}) that there exists a unique
solution to the nonlinear equation (\ref{near-identity-leading-order-zero}) for $A_0 \in \mathbb{C}$
if $\epsilon$ is sufficiently small and $A_1 \in \mathbb{C}$ is given. The solution satisfies the bound
\begin{equation}
\label{bound-A-0}
|A_0| \leq C' |A_1|,
\end{equation}
where the positive constant $C'$ is $\epsilon$-independent. By substituting this solution for $A_0 \in \mathbb{C}$
to the system (\ref{near-identity-leading-order}), we observe that the leading-order system is given by
the recurrence equation
\begin{equation}
-\Omega \bar{A}_n^{(0)} - i\gamma A_n^{(0)} - \bar{A}_{n-1}^{(0)} = E W_n, \quad n \in \mathbb{N},
\label{leading-order-A-n}
\end{equation}
where $A_0^{(0)} = 0$. Since $\Omega > \gamma$, there exists a unique solution $A^{(0)} \in \ell^{\infty}(\mathbb{N})$
of the leading-order system (\ref{leading-order-A-n}). Moreover, because the Jacobian operator $D_x F(0,0)$ is
one-to-one and onto, the solution $A^{(0)}$ is actually in $\ell^2(\mathbb{N})$. By using the implicit function theorem again,
for $\epsilon \neq 0$ sufficiently small, there exists a unique solution $A \in \ell^2(\mathbb{N})$ to the system (\ref{near-identity-leading-order})
satisfying the bound
\begin{equation}
\label{bound-A-n}
\| A - A^{(0)} \|_{\ell^2(\mathbb{N})} < C'' |\epsilon|,
\end{equation}
where the positive constant $C''$ is $\epsilon$-independent.
Combining bounds (\ref{bound-A-0}), (\ref{bound-A-n}) with the representation
(\ref{near-identity-expansion}) yields the bounds (\ref{near-identity-solution}).
\end{proof}

By using the transformation (\ref{near-identity-rho}), we rewrite the expansion (\ref{eq:rearr})
in the following equivalent form
\begin{equation}
\label{eq:rearr-new}
\Delta = \Delta_0 + \Delta_2(\tilde{\phi}) + \Delta_3(\tilde{\phi}) + \Delta_4(\tilde{\phi}),
\end{equation}
where the $\tilde{\phi}$-independent term $\Delta_0$ is given by
$$
\Delta_0 := N_1(\rho) + N_2(\rho) + N_3(\rho) + N_4(\rho),
$$
the quadratic and cubic parts $\Delta_2(\tilde{\phi})$ and $\Delta_3(\tilde{\phi})$ are $\epsilon$-close to
$N_2(\tilde{\phi})$ and $N_3(\tilde{\phi})$, while $\Delta_4(\tilde{\phi}) = N_4(\tilde{\phi})$.
The following proposition characterizes each term of the decomposition (\ref{eq:rearr-new}).
The new definitions of constants override the previous definitions of constants.

\begin{proposition}
\label{proposition-decomposition}
Fix $\gamma > 0$, $\Omega > \gamma$, and $E \neq \pm E_0$.
There exist $\epsilon_0 > 0$ sufficiently small and $C_0, C_1, C_2, C_3, C_4 > 0$ such that
for every $\epsilon \in (-\epsilon_0,\epsilon_0)$, we have
\begin{equation}
|\Delta_0| \leq C_0  \epsilon^2,
\label{eq:lambdatilde}
\end{equation}
\begin{equation}
\| \tilde{\phi}\|_\ltz \le \| \phi \|_\ltz + C_1 \epsilon, \quad
\| \phi \|_\ltz \le \| \tilde{\phi} \|_\ltz + C_1 \epsilon,
\label{bnd:phi}
\end{equation}
\begin{equation}
\label{cubic-quartic-tilde}
\| \Delta_3(\tilde{\phi}) \|_\ltz \leq C_3 \|\tilde{\phi}\|^3_\ltz, \quad
\| \Delta_4(\tilde{\phi}) \|_\ltz \leq C_4 \|\tilde{\phi}\|^4_\ltz,
\end{equation}
and
\begin{equation}
\label{coercivity-tilde}
\Delta_2(\tilde{\phi}) \geq C_2 \| \tilde{\phi} \|_{\ell^2}^2 \quad \mbox{\rm for every} \;\; \tilde{\phi} \in \ltzc.
\end{equation}
\end{proposition}

\begin{proof}
Since $\rho$ is constructed in Proposition \ref{proposition-near-identity} with the $\PT$-symmetric
correction term $b = \bar{a}$, it is true that $\rho \in \ell^2_c(\mathbb{Z})$.
Therefore, the condition $\phi \in \ell^2_c(\mathbb{Z})$ is satisfied if and only if
$\tilde{\phi} \in \ell^2_c(\mathbb{Z})$.
Since the constants $C_2$, $C_3$, and $C_4$ in the bounds (\ref{cubic-quartic}) and (\ref{coercivity}) are
$\epsilon$-independent, whereas $\Delta_2$, $\Delta_3$, and $\Delta_4$ are $\epsilon$-close to
$N_2$, $N_3$, and $N_4$ in space $\ell^2(\mathbb{Z})$, then the bounds
(\ref{cubic-quartic-tilde}) and (\ref{coercivity-tilde}) follow from
the bounds (\ref{cubic-quartic}) and (\ref{coercivity}) respectively,
thanks to the smallness of $\epsilon$.

In order to obtain the bounds (\ref{eq:lambdatilde}) and (\ref{bnd:phi}),
we use (\ref{bound-soliton}) and (\ref{near-identity-solution})
to obtain
\begin{equation}
\label{estimates-above}
|N_1(\rho)| \leq C \sum_{n \in \mathbb{Z} \backslash \{0\}} \epsilon^{2|n|} \leq C'  \epsilon^2,  \quad
\| \rho\|^2_\ltz \leq C \left( \epsilon^4 + \sum_{n \in \mathbb{Z} \backslash \{0\}} \epsilon^{2|n|} \right) \leq C' \epsilon^2,
\end{equation}
where the positive constants $C$, $C'$ are $\epsilon$-independent and $\epsilon$ is sufficiently small.
Since $N_2$, $N_3$, and $N_4$ are quadratic, cubic, and quartic respectively, the bound (\ref{eq:lambdatilde})
is obtained from the triangle inequality and the estimates (\ref{estimates-above}). The bounds
(\ref{bnd:phi}) follow from the triangle inequality and the second estimate (\ref{estimates-above}).
\end{proof}

\subsection{Time evolution of $\Delta$}

We recall that $H$ given by (\ref{energy-PT-dNLS}) is a constant of motion for the $\PT$-symmetric dNLS equation (\ref{PT-dNLS}).
On the other hand, the part of $Q$ at $n = 0$ satisfies the balance equation
\begin{equation}
\label{zero-flux-Delta}
i \frac{d}{dt} \left( u_0 \bar{v}_0 + \bar{u}_0 v_0 \right) = \epsilon \left[
\bar{u}_0 (u_1 + u_{-1}) - u_0 (\bar{u}_1 + \bar{u}_{-1}) + \bar{v}_0 (v_1 + v_{-1}) - v_0 (\bar{v}_1 + \bar{v}_{-1}) \right].
\end{equation}

If the initial data $\psi(0) \in\ltzz$ is close to $\Phi$ in the sense of the bound
$\| \psi(0) - \Phi \|_{\ell^2} \leq \delta$, then the unique solution
$\psi(t) \in C^1(\mathbb{R},\ltzz)$ to the $\PT$-symmetric dNLS equation (\ref{PT-dNLS})
with the same initial data can be defined in the modulation form
\begin{equation}
\psi(t) = e^{-i \alpha(t) \sigma} \left[ \Phi + \phi(t) \right],
\label{eq:psi}
\end{equation}
as long as the solution remains close to the orbit of $\Phi$ under the phase rotation
in the sense of the bound (\ref{bound-final}). Note again that the vectors $\psi$, $\Phi$, and $\phi$ 
are extended $4$-component vectors at each lattice site compared to the $2$-component vectors 
used in the formulation of Theorem \ref{theorem-main}. As a result, the gauge symmetry 
is represented by the matrix operator $\sigma$ defined by (\ref{kernel}).

The decomposition (\ref{eq:psi}) is defined uniquely only if a constraint is imposed to $\phi(t) \in \ell^2(\mathbb{R})$.
In agreement with the definition (\ref{cond:l2c}) on the constrained space $\ell^2_c(\mathbb{R})$, we
impose the orthogonality condition:
\begin{equation}
	\langle \sigma \Phi,\phi(t) \rangle_\ltz = 0.
\label{cond:ortho}
\end{equation}
The decomposition (\ref{eq:psi}) under the orthogonality condition (\ref{cond:ortho})
and the modulation equation for $\alpha$ are justified in
the next section. Here we estimate how the time-dependent energy quantity $\Delta$ changes along
the solution $\psi(t) \in C^1(\mathbb{R},\ltzz)$ represented by the decomposition (\ref{eq:psi}).

The rate of change of $\Delta$ defined by (\ref{eq:convex}) along the solution $\psi(t)$
represented by (\ref{eq:psi}) is obtained from (\ref{zero-flux-Delta}) as follows:
\begin{equation}
\label{zero-flux}
\left| \frac{d\Delta}{dt} \right| \le C_E \epsilon
\| \Phi_0 + \phi_0 \| \left( \|\Phi_1 + \phi_1 \| + \| \Phi_{-1} + \phi_{-1} \| \right)
\end{equation}
where $C_E$ is a positive $\epsilon$-independent constant. By using the bounds (\ref{bound-soliton}) and (\ref{near-identity-solution}),
the transformation (\ref{near-identity-rho}),
and the triangle inequality (\ref{bnd:phi}), we obtain from (\ref{zero-flux}):
\begin{eqnarray}
\nonumber
\left| \frac{d\Delta}{dt} \right| & \le & C_E \epsilon \left( 1 + \| \tilde{\phi}_0 \| \right)
\left( \epsilon + \| \tilde{\phi}_1 \| + \| \tilde{\phi}_{-1} \| \right) \\
& \le & C_E' \epsilon \left( \epsilon + \| \tilde{\phi}_0 \| + \| \tilde{\phi}_1 \| + \| \tilde{\phi}_{-1} \|
+ \| \tilde{\phi}_0 \|^2 + \| \tilde{\phi}_1 \|^2 + \| \tilde{\phi}_{-1} \|^2 \right),
	\label{eq:derbnd}
\end{eqnarray}
where $C_E'$ is another positive $\epsilon$-independent constant.

Let us now define a ball in the space $\ell^2_c(\mathbb{Z})$
of a finite size $K > 0$ by
\begin{equation}
\mathcal{M}_K := \left\{ \phi \in \ell^2_c(\mathbb{Z}) : \quad \| \phi \|_{\ell^2} \leq K \right\}.
\label{ball}
\end{equation}
From estimates (\ref{coercivity-tilde}) and (\ref{cubic-quartic-tilde}), there is a positive $K$-dependent
constant $C_K$ such that
\begin{equation}
\label{coercivity-final}
\Delta - \Delta_0 \geq C_K \| \tilde{\phi} \|_{\ell^2}^2 \quad \mbox{\rm for every} \;\; \tilde{\phi} \in  \mathcal{M}_K.
\end{equation}
By using coercivity (\ref{coercivity-final}) in the ball $\mathcal{M}_K$
and the Young inequality
$$
|a b| \le  \frac{\alpha}{2}  a^2 + \frac{1}{2\alpha} b^2, \quad a,b \in \mathbb{R},
$$
where $\alpha \in \mathbb{R}^+$ is arbitrary, we estimate
$$
\| \tilde{\phi}_0 \| + \| \tilde{\phi}_1 \| + \| \tilde{\phi}_{-1} \|\le \sqrt {C_K^{-1} (\Delta - \Delta_0) } \le
\frac{\alpha}{2C_K} + \frac{1}{2 \alpha} (\Delta - \Delta_0),
$$
where $\Delta - \Delta_0 \geq 0$ follows from (\ref{coercivity-final}).
Substituting this estimate to \eqref{eq:derbnd} yields
\begin{equation}
\left| \frac{d\Delta}{dt} \right| \le C_E \epsilon \left( \epsilon +
\alpha + (\Delta - \Delta_0) + \alpha^{-1} ( \Delta - \Delta_0 ) \right),
	\label{ineq:dlam}
\end{equation}
for another constant $C_E > 0$. In what follows, we will set the scaling
parameter $\alpha$ such that $\alpha \to 0$ as $\epsilon \to 0$.
Therefore, the constant $\alpha^{-1}$ is much larger compared to unity. Integrating (\ref{ineq:dlam})
with an integrating factor,
$$
\left| \frac{d}{dt} e^{-C_E \epsilon \alpha^{-1} t} (\Delta - \Delta_0) \right| \le
C_E \epsilon (\epsilon + \alpha) e^{-C_E \epsilon \alpha^{-1} t},
$$
we obtain with the Gronwall's inequality:
\begin{eqnarray}
\nonumber
\Delta(t) - \Delta_0 & \le & e^{C_E \epsilon \alpha^{-1} t} \left( \Delta(0) - \Delta_0 +
C_E \epsilon (\epsilon + \alpha) \int_0^t e^{-C_E \epsilon \alpha^{-1} s} ds \right) \\
& \le & e^{C_E \epsilon \alpha^{-1} t} \left( \Delta(0) - \Delta_0 +
\alpha (\epsilon + \alpha)\right).
\label{eq:lambd}
\end{eqnarray}
It is clear from the estimate (\ref{eq:lambd}) that $\Delta(t) - \Delta_0$ is small
only if $\alpha \to 0$ as $\epsilon \to 0$. If $\alpha = \epsilon$, then
$$
\alpha = \epsilon : \quad \Delta(t) - \Delta_0 \leq e^{C_E t} \left( \Delta(0) + \epsilon^2 \right),
$$
where the bound (\ref{eq:lambdatilde}) has been used. Therefore, if $\Delta(0)$ is small,
then $\Delta(t)$ remains small on the time scale $t = \mathcal{O}(1)$ as $\epsilon \to 0$.
On the other hand, if $\alpha = \epsilon^{1/2}$, then the estimate (\ref{eq:lambd}) yields
$$
\alpha = \epsilon^{1/2} : \quad \Delta(t) - \Delta_0 \leq e^{C_E \epsilon^{1/2} t} \left( \Delta(0) + \epsilon \right),
$$
so that $\Delta(t)$ remains small on the time scale $t = \mathcal{O}(\epsilon^{-1/2})$.

The initial value for $\Delta(0)$ is estimated from (\ref{eq:rearr-new}), (\ref{eq:lambdatilde}), and (\ref{bnd:phi}).
By (\ref{bnd:phi}), for every $\phi(0) \in \mathcal{M}_{\delta}$ with $\delta > 0$ sufficiently small, we have
$\tilde{\phi}(0) \in \mathcal{M}_K$
with $K = \delta + \epsilon$ and there are positive ($\epsilon$,$\delta$)-independent constants $C,C'$ such that
\begin{equation}
\label{Delta-initial}
|\Delta(0)| \leq C \left( \epsilon^2 + \| \tilde{\phi}(0) \|_{\ell^2}^2 \right) \leq C' (\epsilon^2 + \delta^2).
\end{equation}
By using the triangle inequality (\ref{bnd:phi}), coercivity (\ref{coercivity-final}), and the bound (\ref{Delta-initial}),
we finally obtain the following two estimates:
$$
\alpha = \epsilon : \quad \| \phi(t) \|_{\ell^2}^2 \leq C e^{C_E t} \left( \epsilon^2 + \delta^2 \right)
$$
and
$$
\alpha = \epsilon^{1/2} : \quad \| \phi(t) \|_{\ell^2}^2  \leq C e^{C_E \epsilon^{1/2} t} \left( \epsilon + \delta^2 \right),
$$
where the positive constant $C$ is independent of $\epsilon$ and $\delta$.
Comparing with the bound (\ref{bound-final}) stated in Theorem \ref{theorem-main}, we obtain
\begin{equation}
\label{estimate-parameters-short}
\alpha = \epsilon, \quad t_0 \lesssim 1 : \quad C (\epsilon + \delta) \leq \nu
\end{equation}
and
\begin{equation}
\label{estimate-parameters}
\alpha = \epsilon^{1/2}, \quad t_0 \lesssim \epsilon^{-1/2} : \quad C \left( \epsilon^{1/2} + \delta \right) \leq \nu,
\end{equation}
where $t_0$ is the final time in the bound (\ref{bound-final}) and $C$ is another positive ($\epsilon$,$\delta$)-independent constant.
For every $\nu > 0$, there exist $\epsilon_0 > 0$ and $\delta > 0$ such that
inequalities (\ref{estimate-parameters-short}) and (\ref{estimate-parameters})
can be satisfied for every $\epsilon \in (0,\epsilon_0)$.
The statement of Theorem \ref{theorem-main} is formulated on the extended time
scale corresponding to the inequality (\ref{estimate-parameters}).
The short time scale corresponding to the inequality
(\ref{estimate-parameters-short}) is mentioned in Remark \ref{remark-shorter-scale}.

\subsection{Modulation equations in $\ell^2_c(\mathbb{Z})$}

It remains to show how we can define the decomposition (\ref{eq:psi})
under the constraint (\ref{cond:ortho}) for a solution to the $\PT$-symmetric dNLS equation (\ref{PT-dNLS})
and how the evolution of $\alpha$ in time $t$ can be estimated from the modulation equation.
Here we modify standard results on modulation equations, see, e.g., Lemmas 6.1 and 6.3 in
\cite{GP} for similar analysis. For reader's convenience, we only give the main ideas behind the proofs.

\begin{proposition}
\label{lemma-orbit}
There exist constants $\nu_0 \in (0,1)$ and $C_0\ge 1$ such that,
for any $\psi \in \ltzz$ satisfying
\begin{equation}
d:= \inf_{\alpha\in\mathbb{R}} \| e^{i \alpha \sigma} \psi - \Phi \|_{\ltz} \le \nu_0, \label{eq:inf}
\end{equation}
one can find modulation parameter $\alpha \in\mathbb{R}/(2\pi\mathbb{Z})$
such that $\psi = e^{- i \alpha \sigma}(\Phi + \phi)$  with $\phi \in \ell^2_c(\mathbb{Z})$
satisfying $d \le \| \phi \|_{\ltz} \le C_0 d$.
\end{proposition}

\begin{proof}
We consider a function $f : \mathbb{R} \to \mathbb{R}$ given by
$$
f(\alpha) := \langle \sigma \Phi, e^{i \alpha \sigma} \psi - \Phi \rangle_{\ell^2} = 0.
$$
Let $\alpha_0 \in \mathbb{R}/(2\pi\mathbb{Z})$ be the argument of the infimum in (\ref{eq:inf}).
Then, $|f(\alpha_0)| \leq d \| \Phi \|_{\ell^2}$ by the Cauchy--Schwartz inequality.
On the other hand, the derivative $f'(\alpha_0)$ is bounded away from zero because
$$
f'(\alpha_0) = \langle \sigma \Phi, i \sigma e^{i \alpha_0 \sigma} \psi \rangle_{\ell^2}
= i \| \Phi \|_{\ell^2}^2 + i \langle \Phi, e^{i \alpha_0 \sigma} \psi - \Phi \rangle_{\ell^2},
$$
where the second term is bounded by $d \| \Phi \|_{\ell^2}$ and the first term is $d$-independent.
The function $f : \mathbb{R} \to \mathbb{R}$ is smooth in $\alpha$. By the implicit function theorem,
for any $d > 0$ sufficiently small, there is a unique solution of the equation $f(\alpha) = 0$ for $\alpha$ near $\alpha_0$
such that $|\alpha - \alpha_0| \leq C d$, where $C$ is $d$-independent. By the triangle inequality,
$\|  \phi \|_{\ltz} \le C_0 d$, where $C_0$ is also $d$-independent.
\end{proof}

\begin{proposition}
\label{lemma-modulation}
Assume that the solution $\psi(t)$ to the $\PT$-symmetric dNLS equation (\ref{PT-dNLS})
satisfies $d(t) \leq \nu$ for every $t \in [0,t_0]$, where
$d(t)$ is given by (\ref{eq:inf}). Then the modulation parameter $\alpha(t)$
defined by (\ref{eq:psi}) in Proposition \ref{lemma-orbit} is a
continuously differentiable function of $t$ and there is a positive constant $C$
such that $|\dot{\alpha}-E| \leq C \nu$, for every $t \in [0,t_0]$.
\end{proposition}

\begin{proof}
Let $\psi(t) \in C^1(\mathbb{R},\ltzz)$ be a solution to the $\PT$-symmetric dNLS equation (\ref{PT-dNLS}).
Substituting the decomposition~(\ref{eq:psi}) into the $\PT$-symmetric dNLS equation (\ref{PT-dNLS}),
we obtain the evolution equation in the form
\begin{equation}
i \dot{\phi} =  \mathcal{S} \mathcal{H}''_E \phi + (E-\dot{\alpha}) \sigma (\Phi + \phi) + N(\phi),
\label{eq:evolution}
\end{equation}
where the bounded invertible operator $\mathcal{S} : \ell^2(\mathbb{Z}) \to \ell^2(\mathbb{Z})$
represents the symplectic structure (\ref{Hamiltonian-PT-dNLS}) of the $\PT$-symmetric
dNLS equation (\ref{PT-dNLS}), $N(\phi)$ contains quadratic and cubic terms in $\phi$,
and the gauge invariance of the $\PT$-symmetric dNLS equation (\ref{PT-dNLS}) has been used.
From the condition (\ref{cond:ortho}), projecting the evolution equation (\ref{eq:evolution}) to $\sigma \Phi$ yields
\begin{equation}
\dot{\alpha} - E =
\frac{\langle \mathcal{H}''_E \mathcal{S} \sigma \Phi, \phi \rangle_{\ell^2} + \langle \sigma \Phi, N(\phi) \rangle_{\ell^2}}{
\| \Phi \|_{\ell^2}^2 + \langle \Phi, \phi \rangle_{\ell^2}}. \label{eq:dalpha}
\end{equation}
By Proposition \ref{lemma-orbit}, if $d(t) \leq \nu$ is sufficiently small for every $t \in [0,t_0]$, then $\| \phi(t) \|_{\ell^2} \leq C_0 d(t)$
for a positive constant $C_0$.
Then, the denominator in (\ref{eq:dalpha}) is bounded away from zero, whereas
the numerator is bounded by $C d(t)$, which yields the bound $|\dot{\alpha}-E| \leq C \nu$, for every $t \in [0,t_0]$.
\end{proof}

\appendix
\section{Useful results}

$\phantom{text}$ {\bf Implicit Function Theorem.} (Theorem 4.E in \cite{Zeidler}) {\em
Let $X,Y$ and $Z$ be Banach spaces and let $F(x,y)\colon X\times Y \to Z$
be a $C^1$ map on an open neighborhood of the point $(x_0,y_0)\in X\times Y$.
Assume that
$$
F(x_0,y_0) = 0
$$
and that
$$
D_x F(x_0,y_0) \colon X\to Z \text{ is one-to-one and onto. }
$$
There are $r>0$ and $\sigma>0$ such that for each $y$ with
$\| y - y_0 \|_Y \le \sigma$ there exists a unique solution
$x\in X$ of the nonlinear equation $F(x,y) = 0$ with
$\|x - x_0\|_X \le r$. Moreover, the map $Y\owns y\mapsto x(y)\in X$
is $C^1$ near $y = y_0$.
}

\vspace{0.2cm}

{\bf Perturbation Theory for Linear Operators.} (Theorem VII.1.7 in \cite{Kato}) {\em
Let $T(\epsilon)$ be a family of operators from Banach space $X$ to itself,
which depends analytically on the small parameter $\epsilon$. If the spectrum of $T(0)$
is separated into two parts, the subspaces of $X$ corresponding to the separated parts also depend on $\epsilon$
analytically. In particular, the spectrum of $T(\epsilon)$ is separated into two parts for any $\epsilon \neq 0$
sufficiently small. }

\vspace{0.2cm}

{\bf Lyapunov's Stability Theorem.} \cite{Lyapunov} {\em
Consider the following evolution problem on a Hilbert space $X$,
\begin{equation}
\frac{d\vec{x}}{dt} = \vec{f}(\vec{x}), \quad x\in X, \label{eq:lyapunov}
\end{equation}
where $\vec{f} : X \to X$ satisfies $\vec{f}(\vec{0}) = \vec{0}$.
Let $V\colon X\to\mathbb{R}$ satisfy the following properties:
\begin{enumerate}
\item $V\in C^2(X)$ with $V(\vec{0}) = 0$;
\item There exists $C > 0$ such that $V(\vec{x}) \ge C\|\vec{x}\|^2_X$ for every $\vec{x}\in X$;
\item $\frac{d}{dt} V(\vec{x}) \le 0$ for every solution of (\ref{eq:lyapunov}).
\end{enumerate}
Then the zero equilibrium of the evolution system~\eqref{eq:lyapunov} is nonlinearly stable
in the sense: for every $\nu > 0$ there is $\delta > 0$ such that if $\vec{x}_0 \in X$
satisfies $\| \vec{x}_0 \|_X \leq \delta$, then the unique solution $\vec{x}(t)$ of
the evolution system (\ref{eq:lyapunov}) such that $\vec{x}(0) = \vec{x}_0$
satisfies $\| \vec{x}(t)\|_X \leq \epsilon$ for every $t \in \mathbb{R}^+$.}

\vspace{0.25cm}

\noindent{\bf Acknowledgements.}
The work of A.C. is supported by the graduate scholarship at McMaster University.
The work of D.P. is supported by the Ministry of Education
and Science of Russian Federation (the base part of the State task No. 2014/133, project No. 2839).

\end{document}